\documentclass[aps,pra,twocolumn,superscriptaddress,floatfix,
nofootinbib,showpacs,longbibliography]{revtex4-2}

\usepackage[normalem]{ulem}
\usepackage{braket}
\usepackage{xcolor}
\usepackage[utf8]{inputenc}  
\usepackage[T1]{fontenc}     
\usepackage[british]{babel}  
\usepackage[sc,osf]{mathpazo}
\usepackage[scaled=0.86]{berasans}  
\usepackage[colorlinks=true, citecolor=blue, urlcolor=blue]{hyperref}  
\usepackage{physics}
\usepackage{graphicx} 
\usepackage[babel]{microtype}  
\usepackage{amsmath,amssymb,amsthm,bm,amsfonts,mathrsfs,bbm} 

\usepackage{xspace}  
\definecolor{wrwrwr}{rgb}{0.3803921568627451,0.3803921568627451,0.3803921568627451}
\definecolor{rvwvcq}{rgb}{0.08235294117647059,0.396078431372549,0.7529411764705882}
\usepackage{pgf,tikz,pgfplots}
\usetikzlibrary{calc}
\pgfplotsset{compat=1.15}
\usepackage{mathrsfs}
\usetikzlibrary{arrows,snakes}
\usetikzlibrary{positioning}
\usepackage{xcolor}
\usepackage{appendix}
\usepackage{multirow}
\usepackage{array}
\usepackage{bigstrut}
\usepackage{braket}
\usepackage{color}
\usepackage{natbib}
\usepackage{multirow}
\usepackage{float}
\usepackage[caption = false]{subfig}
\usepackage{xcolor,colortbl}
\usepackage{color}

\newtheorem{theorem}{Theorem}

\newtheorem{proposition}{Proposition}

\begin{document}

\title{Genuine Activation of Quantum Nonlocality: Stronger Than Local Indistinguishability}
	
\author{Tathagata Gupta}
\affiliation{Physics and Applied Mathematics Unit, Indian Statistical Institute, 203 B.T. Road, Kolkata 700108, India.}

\author{Subhendu B. Ghosh}
\affiliation{Physics and Applied Mathematics Unit, Indian Statistical Institute, 203 B.T. Road, Kolkata 700108, India.}

\author{Ardra A V}
\affiliation{School of Physics, IISER Thiruvananthapuram, Vithura, Kerala 695551, India.}

\author{Anandamay Das Bhowmik}
\affiliation{Physics and Applied Mathematics Unit, Indian Statistical Institute, 203 B.T. Road, Kolkata 700108, India.}

\author{Sutapa Saha}
\affiliation{Physics and Applied Mathematics Unit, Indian Statistical Institute, 203 B.T. Road, Kolkata 700108, India.}

\author{Tamal Guha}
\affiliation{Department of Computer Science, The University of Hong Kong, Pokfulam Road, Hong Kong.}

\author{Ramij Rahaman}
\affiliation{Physics and Applied Mathematics Unit, Indian Statistical Institute, 203 B.T. Road, Kolkata 700108, India.}

\author{Amit Mukherjee}
\affiliation{S. N. Bose National Centre for Basic Sciences, Block JD, Sector III, Saltlake, Kolkata 700098, India.}
\begin{abstract}
{The structure of quantum theory assures the discrimination of any possible orthogonal set of states. However, the scenario becomes highly nontrivial in the limited measurement setting and leads to different classes of impossibility, viz., indistinguishability, unmarkability, irreducibility etc. These phenomena, often referred to as other nonlocal aspects of quantum theory, have utmost importance in the domain of secret sharing etc. It, therefore, becomes a pertinent question to distill/activate such behaviours from a set, apparently devoid of these nonlocal features and free from local redundancy. While the activation of local indisitnguishability in the sets of entangled states has recently been reported, other stronger notion of quantum nonlocality has yet not been inspected in the parlance of activation. Here, we explored all such stronger versions of nonlocality and affirmatively answered to activate each of them. We also find a locally distinguishable set of multipartite entangled states which can be converted with certainty to a nontrivial set of locally irreducible genuinely entangled states.}
\end{abstract}
\maketitle
\section{Introduction}
Quantum entanglement \cite{Horodecki2009} is a necessary resource to manifest some remarkable non-classical features. Bell nonlocality \cite{Bell64,Bell66,Mermin93RMP,Brunner14} is perhaps the most celebrated among those counter-intuitive phenomena. Apart from its profound foundational perspective Bell nonlocality also acts as an essential resource in a plethora of intriguing applications \cite{nsqkd,btsqkd,random,renner12,testdim,postqbanik,dvdim,mdran,gamepappa,Roy2016,game2}. However, Bell nonlocality is not the sole member of this counter-intuitive non-classical club. Quantum mechanics also empowers a number of other nonlocal phenomena, such as, local indistinguishability \cite{walgate,lid3} and irreducibiity \cite{halderprl}, impossibility of local state marking \cite{lsm} etc. All those tasks show some nonlocal feature different from the one named after J. S. Bell. Here, we will particularly focus on these three kinds of nonlocal task.

It is well known that identifying a multipartite quantum state from an orthogonal set by some spatially separated players is a nontrivial task whenever they have restricted measurement settings. It may not always be possible to accomplish this task even when the players are allowed to perform local operations along with classical communication among each other. This surprising fact is usually termed as local indistinguishability of quantum states. Since its inception, this phenomenon has fueled a vast amount of studies \cite{lid1,lid2,lid3,lid4,lid5,lid6,lid7,lid8,lid9,lid11,lid12,lid13,lid14,lid15,lid16,lid17,lid18,lid19,lid20,lid21,lid22,lid23,lid24,lid25,lid26,lid28,lid29,lid30,lid31,lid32,lid33,lid34,lid35,lid36,lid37,lid38,lid39,halderpra,halderpra2,halderprl,lsm}. Another compelling study in this area is recently introduced by Halder \textit{et al.} - local irreducibility \cite{halderprl}. Their question involves elimination of some state(s) from a set of multipartite states via LOCC while keeping the orthogonality relation among the updated states intact. As an instance, the set of four two-qubit Bell states which is a well known locally indisitnguishable set has also been shown as locally irreducible. Very recently, Sen \textit{et al.} introduced another novel task of correctly marking the members of a set of orthogonal  multipartite states shared between spatially separated players \cite{lsm}. The players are restricted to perform LOCC only. In this setting, it has been shown that the above mentioned set of four Bell states are locally unmarkable.

Emergence of a number of nonlocal features also raises a pertinent question - which are stronger than others? It has been demonstrated that local indisitinguishability of an orthogonal set does not necessarily guarantee that the set must also be locally irreducible or unmarkable. Rather, the converse is proved to be true. Any set of states which is locally irreducible must also be locally indistinguishable \cite{halderprl}. In addition to that it has also been established that local unmarkability readily implies local indistinguishability \cite{lsm}.
Naturally, both the notions of local irreducibility and unmarkability are stronger nonlocal phenomena than local indistinguishability. At this point, it may appear that quantum entanglement is the sole reason behind such interesting nonlocal properties. Surprisingly, there exist sets of product states which are also irreducible under LOCC \cite{bennps,cmp,upb,halderprl}. This striking phenomenon is usually termed as `quantum nonlocality without entanglement'. However, here we are only interested in various nonlocal features manifested in entangled states.

The nonlocal features mentioned above are of immense importance in various quantum security schemes, such as, data hiding \cite{DiViieee,terprl,wernerprl}, secret sharing \cite{markhampra} etc. Thus, it is a question of utmost importance to transform an orthogonal set of non-resource states to a nonlocal set by local orthogonality preserving measurements (OPM) only. Bandyopadhyay and Halder \cite{sar} recently provided some examples in this direction. They introduced a number of locally distinguishable sets which can be deterministically converted to locally indistinguishable ones just by performing some local OPMs. This phenomenon has been called \textit{genuine}\footnote{The authors of \cite{sar} termed this activation `genuine' to ensure that the initial orthogonal set cannot be trivially converted to some nonlocal sets just by reducing some subsystem(s).} \textit{activation of nonlocality}. Their work has already motivated a number of other studies on activation of quantum nonlocality without entanglement \cite{liz,our}.   The work \cite{liz} deals with the activation of locally irreducible set from locally disitinguishable sets of bipartite product states. Some authors of this current article have very recently explored the activation of strongest possible form \cite{halderprl,gnps,minimal} of quantum nonlocality without entanglement in sets of multipartite orthogonal product states \cite{our}. However, none of those previous works delve into the activation of above mentioned stronger nonlocal features in sets of multipartite entangled states.
In this article, we explore this question. Here, we introduce a number of locally distinguishable sets that can be deterministically transformed, via local OPMs, into sets of various stronger nonlocal characteristics. Precisely, we present more than one locally distinguishable sets which can be transformed into locally unmarkable sets with certainty.  Furthermore, we also provide a set of multipartite states which can be deterministically activated to a locally irreducible set of genuinely entangled states with cardinality more than three.

The rest of the paper is structured as follows. In section \ref{sec:lsm}, we delve into the activation of a stronger nonlocal property, namely, local unmarkability. In section \ref{sec:ghz}, we explore the activation of another strong nonlocal feature - local irreducibility in multipartite genuinely entangled states. Lastly, section \ref{sec:dis} contains our concluding remarks and possible open directions in this area.

\section{Activation of local unmarkability}\label{sec:lsm}
In this section, we will present a number of orthogonal sets in support of our assertion. Precisely speaking, we will introduce few examples of locally distinguishable sets, from which it is possible to activate several possible notions of indistinguishability. Before going to the main theorems we propose an orthogonal set of bipartite states, $\mathcal{S}_1\equiv\{\ket{\psi_i}_{AB}\}_{i=1}^4 \subset \mathbb{C}^2\otimes \mathbb{C}^4$, where,
\begin{subequations}\label{saro}
\begin{eqnarray}
\ket{\psi_1}_{AB}&=&\ket{0\mathbf{0}}+\ket{0\mathbf{2}}+\ket{1\mathbf{1}}-\ket{1\mathbf{3}}\\
\ket{\psi_2}_{AB}&=&\ket{0\mathbf{0}}-\ket{0\mathbf{2}}-\ket{1\mathbf{1}}-\ket{1\mathbf{3}}\\
\ket{\psi_3}_{AB}&=&\ket{0\mathbf{1}}-\ket{1\mathbf{2}}-\ket{1\mathbf{0}}-\ket{0\mathbf{3}}\\
\ket{\psi_4}_{AB}&=&\ket{0\mathbf{1}}-\ket{1\mathbf{2}}+\ket{1\mathbf{0}}+\ket{0\mathbf{3}}.
\end{eqnarray}
\end{subequations}
Here, we represent the ququad system of Bob as a composition of two qubits. In our construction, $\ket{\mathbf{0}}_B:=\ket{00}_{b_1b_2},\ket{\mathbf{1}}_B:=\ket{01}_{b_1b_2},\ket{\mathbf{2}}_B:=\ket{10}_{b_1b_2},\ket{\mathbf{3}}_B:=\ket{11}_{b_1b_2}.$
\begin{proposition}\label{4s}
The set $\mathcal{S}_1$ is locally distinguishable and free from local redundancy.
\end{proposition}
\begin{proof}

First we will describe the local distinguishability protocol of the set $\mathcal{S}_1$. Suppose, Alice performs a measurement $\mathcal{M}_A\equiv\{{M}_1^A:=P[\ket{0}_A],{M}_2^A:=P[\ket{1}_A]\}$. Here, $P[(\ket{i},\ket{j})_\#]:=(\ket{i}\bra{i}+\ket{j}\bra{j})_\#$, and $\#$ denotes the party. Regardless of Alice's outcome, Bob measures in the basis $\{(\ket{\mathbf{0}}\pm\ket{\mathbf{2}}),(\ket{\mathbf{1}}\pm\ket{\mathbf{3}})\}$. Then, by communicating their results, they can identify the state. 

To check the condition of local redundancy, observe that, discarding the Bob's first qubit reduces both $\ket{\psi_{3}}_{AB}$ and $\ket{\psi_{4}}_{AB}$ to $\frac{1}{2}(\ketbra{0}{0}_{A}\otimes\ketbra{1}{1}_{b_{2}}+\ketbra{1}{1}_{A}\otimes\ketbra{0}{0}_{b_{2}})$ , while discarding the second qubit produces $\frac{1}{2}(\ketbra{0}{0}_{A}\otimes\ketbra{-}{-}_{b_{1}}+\ketbra{1}{1}_{A}\otimes\ketbra{+}{+}_{b_{1}})$ from both $\ket{\psi_{2}}_{AB}$ and $\ket{\psi_{3}}_{AB}$. On the other hand, discarding Alice's qubit reduces both $\ket{\psi_{2}}_{AB}$ and $\ket{\psi_{4}}_{AB}$ to an uniform ensemble of $\ket{+}_{b_{1}}\ket{1}_{b_{2}}$ and $\ket{-}_{b_{1}}\ket{0}_{b_{2}}$ and hence this completes the proof.
\end{proof}
The possible notion of nonlocality activation from the set $\mathcal{S}_{1}$ will be depicted later. Before that, let us consider another set $\mathcal{S}_2\equiv\{\ket{\xi_i}_{AB}\}_{i=1}^4 \subset \mathbb{C}^4\otimes \mathbb{C}^8$, where, \begin{subequations}
\begin{eqnarray}
\ket{\xi_1}_{AB}&=&\ket{\psi_1}_{A_1B_1}\otimes \ket{\phi^+}_{A_2B_2}\\
\ket{\xi_2}_{AB}&=&\ket{\psi_2}_{A_1B_1}\otimes \ket{\phi^-}_{A_2B_2}\\
\ket{\xi_3}_{AB}&=&\ket{\psi_3}_{A_1B_1}\otimes \ket{\phi^-}_{A_2B_2}\\
\ket{\xi_4}_{AB}&=&\ket{\psi_4}_{A_1B_1}\otimes \ket{\phi^-}_{A_2B_2}.
\end{eqnarray}
\end{subequations}
Here $\ket{\phi^\pm}$ denote the Bell states $\frac{1}{\sqrt{2}}(\ket{00}\pm\ket{11}$.\\
To emphasize the trivial nonlocality activation from a local redundant set, Bandopadhyay and Halder have come up with an example by composing a locally indistiguishable set with a distinguishable one, such that the joint states remain distinguishable \cite{sar}. Now, by discarding the distinguishable part, one may claim to activate "nonlocality", which is trivial by the construction. Hence to remove such an ambiguity, the authors in \cite{sar}, claims that a local-redundancy free set should lose its orthogonality after discarding any of its subsystems(s). Being a sufficient condition, it not only excludes all such trivial constructions, but also several possible class of nontrivial structures. Here, we modify the criterion slightly, to address only the (more appropriate) issue of local distinguishability. 
Precisely, we want 
to identify a set of distinguishable states as locally redundant, if it becomes indistinguishable just by simple subsystem(s) reduction. To make our point explicit, consider the set $\mathcal{S}_{2}$ which is orthogonal after a particular subsystem reduction, but remains locally distinguishable. However, the set is a potential candidate to activate nonlocality under properly chosen local OPM. So, nonlocality activation from this set is not a trivial claim.
\begin{proposition}
The set of states $\mathcal{S}_2$ is locally distinguishable and reduction of any of its subsystem(s), either makes it nonorthogonal or, a locally distinguishable set.
\end{proposition}
\begin{proof}
The protocol for locally distinguishing the set is quite straightforward. Alice and Bob need to follow the same scheme elucidated in Proposition \ref{4s} for the bipartite subsystems $A_1B_1$. This scheme perfectly distinguishes the set $\mathcal{S}_2$ via LOCC.

We will next move to the proof that $\mathcal{S}_2$ does not have local redundancy. 
We will prove this case by case.\\
\textbf{Case I:} First consider that both Alice and Bob discard their subsystems $A_1B_1$. Consequently, the four resulting states $\{\ket{\zeta`_1}:=\ket{\phi^+}_{A_2B_2},\ket{\zeta`_2}:=\ket{\phi^-}_{A_2B_2},\ket{\zeta`_3}:=\ket{\phi^-}_{A_2B_2},\ket{\zeta`_4}:=\ket{\phi^-}_{A_2B_2}\}$ are not mutually orthogonal.

Further, it follows from the proof of Proposition \ref{4s} that discarding any subsystem(s) from $A_1B_1$ the remaining reduced sets in ${\mathbb{C}^2}\otimes {\mathbb{C}^8}$ or ${\mathbb{C}^4}\otimes{\mathbb{C}^4}$ are also nonorthogonal. 
\\
\textbf{Case II:} Let us now consider that Alice and Bob discard the subsystems $A_2B_2$. In this case, it is evident that the reduced states $\{\ket{\psi_i}_{A_1B_1}\}_{i=1}^4$
are orthogonal and locally distinguishable (see Proposition \ref{4s}). So obviously, this operation does not make the set nonlocal.

Similarly, it is also evident that whenever any of the players discards $A_2$ or $B_2$ the set of reduced states remains locally distinguishable in $A_1B_1$. 

This completes our proof. 
\end{proof}
The orthogonal set $\mathcal{S}_2$ must be locally markable as it is perfectly distinguishable under LOCC. We are, therefore, now in a position to state our first result that the set $\mathcal{S}_2$ can be activated to a nonlocal, i.e., locally indistinguishable set, which is, however, locally markable. Thus, this is an example of weak activation.  

\begin{theorem}
The set $\mathcal{S}_2$ can be deterministically converted to a locally indistinguishable set of states via a local OPM although the final set of states are still locally markable.  
\end{theorem}
\begin{proof}
Suppose, Bob performs a local OPM on subsystem $B_1$,  $\mathcal{N}_{B_1}\equiv\{{N}_1^{B_1}:=P[(\ket{\mathbf{0}},\ket{\mathbf{1}})_{B_1}],{N}_2^{B_1}:=P[(\ket{\mathbf{2}},\ket{\mathbf{3}})_{B_1}]\}$. Clearly, when ${N}_1^{B_1}$ clicks, the players are left with any of the following four orthogonal states in ${\mathbb{C}^2}^{\otimes2}\otimes{\mathbb{C}^2}^{\otimes2}$:
\begin{subequations}\label{duan}
\begin{eqnarray}
(\ket{0\mathbf{p}}+\ket{1\mathbf{q}})_{A_1B_1}&\otimes& \ket{\phi^+}_{A_2B_2}\\
(\ket{0\mathbf{p}}-\ket{1\mathbf{q}})_{A_1B_1}&\otimes& \ket{\phi^-}_{A_2B_2}\\
(\ket{0\mathbf{q}}+\ket{1\mathbf{p}})_{A_1B_1}&\otimes& \ket{\phi^-}_{A_2B_2}\\
(\ket{0\mathbf{q}}-\ket{1\mathbf{p}})_{A_1B_1}&\otimes& \ket{\phi^-}_{A_2B_2}
\end{eqnarray}
\end{subequations}
where, $\mathbf{p}=\mathbf{0}$ and $\mathbf{q}=\mathbf{1}$. Otherwise, when ${N}_2^{B_1}$ clicks, the system belongs to ${\mathbb{C}^2}^{\otimes2}\otimes{\mathbb{C}^2}^{\otimes2}$. Bob performs a phase flip operation on the qubit system $B_1$. Consequently, the state that they are left with is one of four orthogonal states in (\ref{duan}) with $\mathbf{p}=\mathbf{2}$ and $\mathbf{q}=\mathbf{3}$. Furthermore, the orthogonal set in (\ref{duan}) is known to be locally indistinguishable \cite{lid18}. However, those states are perfectly markable under LOCC \cite{lsm}. It, therefore, shows an activation of weaker nonlocality.
\end{proof}
However, an example considered in \cite{sar} can be used to further activate a relatively stronger notion of nonlocality, which we will elucidate in the following theorem. 

\begin{theorem}
A locally distinguishable (hence, locally markable) set $\mathcal{S}_3$ that consists of any three states of (\ref{saro}) can be converted, with certainty, via local OPM to a locally indistinguishable set which is also unmarkable under one-way LOCC.
\end{theorem}
\begin{proof}
In \cite{sar}, the authors proved that the set $S_3$ is locally distinguishable and free from local redundancy. Moreover, under local OPM these set can be deterministically transformed to an orthogonal set consists of three two-qubit Bell states. Naturally, this set is locally indistinguishable \cite{walgate}. Furthermore, Sen \textit{et al.} \cite{lsm} proved that the set $\mathcal{S}_3$ is also an unmarkable one when the players are allowed to perform only one-way LOCC. However, their result does not claim its unmarkability under general LOCC. 
\end{proof}
It is, therefore, a natural question whether there exists any orthogonal set from which strongest form of quantum nonlocality in this direction can be activated. We answer this in affirmation.  
\begin{theorem}\label{th:b1}
The set $\mathcal{S}_1$ can be deterministically converted to a locally unmarkable set of states via a local OPM.  
\end{theorem}
\begin{proof}
Suppose, any of the states from set $\mathcal{S}_1$ is shared between Alice and Bob. Bob then performs a local binary measurement on subsystem $B$, $\mathcal{N}_B\equiv\{{N}_1^B:=P[(\ket{\mathbf{0}},\ket{\mathbf{1}})_B],{N}_2^B:=P[(\ket{\mathbf{2}},\ket{\mathbf{3}})_B]\}$. Whichever projector clicks, they are left with any of four updated states: $\{\ket{0\mathbf{0}}\pm\ket{1\mathbf{1}}, \ket{0\mathbf{1}}\pm\ket{1\mathbf{0}}\}$ or $\{\ket{0\mathbf{2}}\pm\ket{1\mathbf{3}}, \ket{1\mathbf{2}}\pm\ket{0\mathbf{3}}\}$. Both are sets of four maximally entangled states in $\mathbb{C}^2\otimes \mathbb{C}^2$, which is known to be locally unmarkable \cite{lsm}.  
\end{proof}
As the orthogonal set $\mathcal{S}_1$ is shown to be a locally distinguishable one and free from local redundancy (see Proposition \ref{saro}), this is clearly an example of activating strong nonlocality, that is, local indistinguishablity along with local unmarkability. Next, we provide another example of such an extreme activation. We consider an orthogonal set $\mathcal{S}_4:=\{\ket{\xi_i}\}_{i=1}^8$ of eight bipartite states in $\mathbb{C}^4\otimes \mathbb{C}^4$, where, 
\begin{subequations}
\begin{eqnarray}
\ket{\xi_1}_{AB}&=&\ket{\mathbf{0}\mathbf{0}}+\ket{\mathbf{0}\mathbf{2}}+\ket{\mathbf{3}\mathbf{1}}-\ket{\mathbf{3}\mathbf{3}}\\
\ket{\xi_2}_{AB}&=&\ket{\mathbf{0}\mathbf{0}}-\ket{\mathbf{0}\mathbf{2}}-\ket{\mathbf{3}\mathbf{1}}-\ket{\mathbf{3}\mathbf{3}}\\
\ket{\xi_3}_{AB}&=&\ket{\mathbf{0}\mathbf{1}}-\ket{\mathbf{3}\mathbf{2}}-\ket{\mathbf{3}\mathbf{0}}-\ket{\mathbf{0}\mathbf{3}}\\
\ket{\xi_4}_{AB}&=&\ket{\mathbf{0}\mathbf{1}}+\ket{\mathbf{3}\mathbf{2}}+\ket{\mathbf{3}\mathbf{0}}-\ket{\mathbf{0}\mathbf{3}}\\
\ket{\xi_5}_{AB}&=&\ket{\mathbf{1}\mathbf{0}}+\ket{\mathbf{1}\mathbf{2}}+\ket{\mathbf{2}\mathbf{1}}-\ket{\mathbf{2}\mathbf{3}}\\
\ket{\xi_6}_{AB}&=&\ket{\mathbf{1}\mathbf{0}}+\ket{\mathbf{1}\mathbf{2}}-\ket{\mathbf{2}\mathbf{1}}+\ket{\mathbf{2}\mathbf{3}}\\
\ket{\xi_7}_{AB}&=&\ket{\mathbf{1}\mathbf{1}}-\ket{\mathbf{2}\mathbf{2}}+\ket{\mathbf{2}\mathbf{0}}+\ket{\mathbf{1}\mathbf{3}}\\
\ket{\xi_8}_{AB}&=&\ket{\mathbf{1}\mathbf{1}}-\ket{\mathbf{2}\mathbf{2}}-\ket{\mathbf{2}\mathbf{0}}-\ket{\mathbf{1}\mathbf{3}}
\end{eqnarray}
\end{subequations}
Here, the bases of each subsystem is represented as a composite qubit-qubit system. More precisely, $\ket{\mathbf{0}}:=\ket{00},\ket{\mathbf{1}}:=\ket{01},\ket{\mathbf{2}}:=\ket{10},\ket{\mathbf{3}}:=\ket{11}$. 

\begin{proposition}
The set $\mathcal{S}_4:=\{\ket{\psi_i}\}_{i=1}^8$ is locally distinguishable and free from local redundancy.
\end{proposition}
\begin{proof}
We start by describing the local distinguishability protocol.
First, Alice performs a measurement $\mathcal{R}_A\equiv\{{R}_1^A:=P[(\ket{\mathbf{0}},\ket{\mathbf{3}})_A],{R}_2^A:=P[(\ket{\mathbf{1}},\ket{\mathbf{2}})_A]\}$. The post measurement states are $\{\ket{\xi_i}\}_{i=1}^4$ or, $\{\ket{\xi_i}\}_{i=5}^8$ for the first and second outcome respectively. However, it is evident that those two sets of qubit-ququad states are in one to one correspondence to the set $\mathcal{S}_1$. The rest of the distinguishability protocol is thus a straightforward extension of the protocol provided in Proposition \ref{saro}.   

It is also straightforward to show that the set does not have local redundancy. 
If both of the parties discard one of their subsystems they will remain with eight states in $\mathbb{C}^{2}\otimes\mathbb{C}^{2}$, which can not be orthogonal in any way. Further, discarding Bob's first (second) qubit the states $\ket{\xi_{3}}$ and $\ket{\xi_{4}}$ will produce a uniform ensemble of $\ket{001}$ and $\ket{110}$ ($\ket{00-}$ and $\ket{11+}$). Similarly, discarding any of the Alice's qubit will produce $\frac{1}{2}(\ketbra{0}{0}\otimes\ketbra{-}{-}\otimes\ketbra{1}{1}+\ketbra{1}{1}\otimes\ketbra{+}{+}\otimes\ketbra{0}{0})$ both from $\ket{\xi_{3}}$ and $\ket{\xi_{4}}$.
\end{proof}
\begin{theorem}\label{th:b2}
The set $\mathcal{S}_4$ can be deterministically converted to a locally unmarkable set of states via local OPMs.
\end{theorem}
\begin{proof}
In order to activate the nonlocality of the orthogonal set $\mathcal{S}_4$, Bob performs a local OPM, $\mathcal{N}_B\equiv\{{N}_1^B:=P[(\ket{\mathbf{0}},\ket{\mathbf{1}})_B],{N}_2^i:=P[(\ket{\mathbf{2}},\ket{\mathbf{3}})_B]\}$. On the other side, Alice measures, $\mathcal{R}_A\equiv\{{R}_1^A:=P[(\ket{\mathbf{0}},\ket{\mathbf{3}})_A],{R}_2^A:=P[(\ket{\mathbf{1}},\ket{\mathbf{2}})_A]\}$. For all possible outcomes of these two measurements, it is quite evident that the players are left with any of the four two-qubit Bell states which is known to be locally unmarkable under general LOCC \cite{lsm}. This completes our proof. 
\end{proof}
Note that, Theorem \ref{th:b1} and \ref{th:b2} prove that there exists such orthogonal sets from which a stronger form of quantum nonlocality - local unmarkability can be genuinely activated. However, as we mentioned earlier there are other stronger nonlocal features, such as, local irreducibility \cite{halderprl}. The authors of \cite{sar} dealt with a number of examples where the ultimate activated nonlocal (\textit{i.e.}, locally indistinguishable) sets have cardinality three. However, such sets are trivially irreducible under LOCC \cite{halderprl}. In our case, Theorem \ref{th:b1} and \ref{th:b2} provide examples where the activated sets consist of four two-qubit Bell states. Interestingly, these four states are known to be nontrivial example of locally irreducible set \cite{halderprl}. Our example, therefore, not only activates local unmarkability but also the local irreducibility. However, in the above mentioned cases we dealt with bipartite examples only. A pertinent question in this direction would be to activate stronger nonlocal features in multiparty scenario. The next section delves into this particular question. 

\section{activation of local irreducibility in genuine entanglement}\label{sec:ghz}

Let us consider a set of N-partite states $\mathcal{{S}}_{5}^{(N)}:=\{\ket{\eta_{k}(\pm)}\}_{k=0}^{\alpha_{N}}\in\mathbb{C}^{4}\otimes\mathbb{C}^{2^{\otimes(N-1)}}$, such that
\begin{equation}\label{eq5}
  \ket{\eta_{k}(\pm)}=\ket{\mathbf{0},k}\pm\ket{\mathbf{1}, (\alpha_{N}-k)}\pm[\ket{\mathbf{2},k}\mp\ket{\mathbf{3}, (\alpha_{N}-k)}]  
\end{equation}
    where, $\alpha_{p}:=2^{(p-1)}-1$, $k$ is the decimal equivalent of the corresponding $(N-1)$-bit string and $\{\mathbf{0},\mathbf{1},\mathbf{2},\mathbf{3}\}$ represents the composition of two qubits $\{00,01,10,11\}$ respectively. However, for better clarity please see the \textit{Appendix} \ref{ap1} for $N=3$ case.  In the following, we will propose the salient features of this set from the perspective of local discrimination.
\begin{proposition}\label{prop4}
The set $\mathcal{S}_{5}^{(N)}$ is distinguishable under LOCC and free from local redundancy.
\end{proposition}
\begin{proof}
To discriminate this ensemble, all but the first party will measure their respective qubits in  $\mathcal{M}_i\equiv\{{M}_1^i:=P[\ket{0}_i],{M}_2^i:=P[\ket{1}_i]\}$ and via classical communication they can identify the decimal index $k$ and correspondingly the state will be either $\ket{\eta_{k}(\pm)}$ or, $\ket{\eta_{\alpha_{N}-k}(\pm)}$. Now, performing a four-outcome measurement  $\mathcal{N}_A\equiv\{N_1^A:=P[\ket{\mathbf{0+2}}_A],N_2^A:=P[\ket{\mathbf{0-2}}_A],N_3^A:=P[\ket{\mathbf{1+3}}_A],N_4^A:=P[\ket{\mathbf{1-3}}_A]\}$ on the first party's possession, they can discriminate the state perfectly. Precisely,  the outcomes of the first party's measurement correspond to the states $\ket{\eta_{k}(\pm)}$ and $\ket{\eta_{\alpha_{N}-k}(\mp)}$ respectively, for every values of $k$.

Now, we will argue that the set of states is free from local redundancy. Note that, if the first (second) qubit of the first party (personified as Alice) is discarded then the reduced density matrix for each pair of states $\ket{\eta_{k}(\pm)}$ ($\{\ket{\eta_{k}(\pm)},\ket{\eta_{\alpha_{N}-k}(\mp)}\}$) will be identical. Further, if the second party discarded their qubit then the reduced system for each of the pairs $\{\ket{\eta_{k}(\pm)},\ket{\eta_{k+\alpha_{N-1}+1}(\pm)}\}$ will be identical and similar argument runs for all the $(N-1)$ parties, due to the party-symmetric nature of these states.  

This completes the proof.
\end{proof}
Our next theorem will depict the possibility for stronger form of nonlocality activation from the set $\mathcal{S}_{5}^{(N)}$.

\begin{theorem}\label{theo5}
The set $\mathcal{S}_{5}^{(N)}$ can be deterministically converted to a class of locally irreducible $N$-partite genuinely entangled states, which are even indistinguishable when all but one party come together.
\end{theorem}
\begin{proof}
To activate the stronger nonlocality, i.e., local irreduicibility Alice will perform a measurement $\mathcal{N}_{A}\equiv\{N_{1}^{A}:= P[(\ket{\mathbf{0}},\ket{\mathbf{1}})_{A}], N_{2}^{A}:= P[(\ket{\mathbf{2}},\ket{\mathbf{3}})_{A}]\}$ on her possession. For each of her clicks the set $\mathcal{S}_{5}^{(N)}$ will be converted to a set of genuinely entangled states in $\mathbb{C}^{2^{\otimes N}}$. These states can be represented as,
\begin{equation}
\ket{\phi_{k}(\pm)}= \ket{\textbf{p},k}\pm\ket{\textbf{q},(\alpha_{N}-k)}
\end{equation}
where, $\ket{\mathbf{p}}=\ket{\mathbf{0}}~(\ket{\mathbf{2}})$ and $\ket{\mathbf{q}}=\ket{\mathbf{1}}~(-\ket{\mathbf{3}})$ when the projector $N_{1}^{A}~(N_{2}^{A})$ clicks.

Clearly, these states are $N$-partite genuinely entangled GHZ states, upto a local unitary ($\ket{\mathbf{2}}\to\ket{\mathbf{0}};~-\ket{\mathbf{3}}\to\ket{\mathbf{1}}$) for the click $N_{2}^{A}$. Interestingly, these states are locally irreducible \cite{halderprl}.
Now, we will show that these states are locally indistinguishable even if all but one parties collaborate. 

It is easy to see that any $N$-qubit GHZ state, in any $1:(N-1)$ bipartition, can be written as
\begin{equation}\label{eq7}
\ket{\mathcal{G}_{k}(\pm)}= \ket{0}\ket{k}\pm\ket{1}\ket{(\alpha_{N}-k)},
\end{equation}
where, $k\in\{0,\cdots, \alpha_{N}\}$ represents the decimal equivalent of $(N-1)$-bits. Also note that all these $(N-1)$-partite states $\{\ket{k}\}$ (and hence $\{\ket{(\alpha_{N}-k)}\}$) are mutually orthogonal.

Now, it is evident that when all the $(N-1)$-parties come together, they can construct the unitary $\mathbb{U}_{k}^{\pm}$, which takes $\ket{0}^{\otimes(N-1)}\to\ket{k}$ and $\ket{1}^{\otimes(N-1)}\to\pm\ket{(\alpha_{N}-k)}$. Therefore, one can write 
\begin{equation}\label{eq8}
    \ket{\mathcal{G}_{k}(\pm)}=(\mathbb{I}\otimes\mathbb{U}_{k}^{\pm})\ket{\mathcal{G}_{0}(+)}
\end{equation}
 Further, by noting that the states in Eq. (\ref{eq7}) are entangled in $\mathbb{C}^{2}\otimes\mathbb{C}^{2^{\otimes(N-1)}}$ and using Eq. (\ref{eq8}), we can assure that these states can only be distinguished under LOCC with a probability not more than $\frac{1}{2}$ \cite{lid8}. This proves that the states are indistinguishable even in one \textit{vs.} $(N-1)$ bipartition.
\end{proof}

\section{Discussion}\label{sec:dis}
In summary, we have explored the stronger notion of nonlocality activation, starting from the sets of entangled states, perfectly discriminable under LOCC. In a more precise note, we have constructed explicit examples of locally distinguishable set of states in $\mathbb{C}^{2}\otimes\mathbb{C}^{4}$ and $\mathbb{C}^{4}\otimes\mathbb{C}^{4}$ which can be transformed to a set of locally unmarkable, and hence obviously locally indistinguishable, states under the action of orthogonality preserving local measurement. These two different examples suggest that the activation is possible both with identical, as well as reduced cardinality than that of the initial set. Further, to emphasize the difference between the activation of unmarkability and indistinguishability, we have dealt with two other locally distinguishable sets in $\mathbb{C}^{4}\otimes\mathbb{C}^{8}$ and $\mathbb{C}^{2}\otimes\mathbb{C}^{4}$, among which the first one can be transformed to a locally indistinguishable but markable set, while only one-way local unmarkability can be activated from the latter. We have also shown that it is possible to activate another stronger notion of genuine indistinguishability, i.e., local irreduiciblity, from a locally distinguishable set in $\mathbb{C}^{4}\otimes\mathbb{C}^{2^{\otimes (N-1)}}.$

Besides exploring a vivid range of quantum nonlocality from the perspective of state discrimination, our work opens up a number of new directions for future research. Although, the activation of indistinguishability and irreduicibility for the product states have been reported very recently \cite{liz, our}, the notion of unmarkability activation has not been studied yet. Further, it will be interesting to study the possibilities of all such activation for the non-maximally entangled states, which, in turn, may answer towards the activation of strongest possible state indistinguishability with two copies under adaptive LOCC \cite{Banik2021prl}. Another important direction is to consider the possibilities of nonlocality activation under the SEP \cite{Sep1,lid18,Sep2,Sep3} and PPT preserving operations \cite{Ppt1}, which are known to be broader than the set of LOCC operations.

\section{Acknowledgement}
We would like to acknowledge stimulating discussions with Guruprasad Kar, Manik Banik and Saronath Halder. Tamal Guha would like to acknowledge his academic visit at Indian Statistical Institute, Kolkata, during November-December of 2021.

\bibliography{SR}

\onecolumngrid 
\appendix
\onecolumngrid 
\section{Activation of Local Irreducibility: An Example with the Set of Three-Qubit GHZ Basis} \label{ap1}
To get a simple realization of Proposition \ref{prop4} and Theorem \ref{theo5}, let us consider the $N=3$ case. In that case the set of \textit{eight} states $\mathcal{S}_{5}^{(3)}\in\mathbb{C}^{4}\otimes\mathbb{C}^{2^{\otimes2}}$, as depicted in Eq. (\ref{eq5}), will take the form

\begin{eqnarray}\nonumber
\ket{\eta_0(\pm)}&=&\ket{\mathbf{0}0}\pm\ket{\mathbf{1}3}\pm\left(\ket{\mathbf{2}0}\mp\ket{\mathbf{3}3}\right)\\\nonumber
\ket{\eta_1(\pm)}&=&\ket{\mathbf{0}1}\pm\ket{\mathbf{1}2}\pm\left(\ket{\mathbf{2}1}\mp\ket{\mathbf{3}2}\right)\\\nonumber
\ket{\eta_2(\pm)}&=&\ket{\mathbf{0}2}\pm\ket{\mathbf{1}1}\pm\left(\ket{\mathbf{2}2}\mp\ket{\mathbf{3}1}\right)\\\nonumber
\ket{\eta_3(\pm)}&=&\ket{\mathbf{0}3}\pm\ket{\mathbf{1}0}\pm\left(\ket{\mathbf{2}3}\mp\ket{\mathbf{3}0}\right)
\end{eqnarray}
where, each of the states $\{\ket{0},\ket{1},\ket{2},\ket{3}\}$ are the decimal equivalents of two qubits ($b$ and $c$)possessed by the second and third parties (personified as Bob and Charlie respectively), which is revealed by measuring their individual qubits in the computational basis $\{\ket{0},\ket{1}\}$ and communicating the outcomes. Also note that the states $\{\ket{\mathbf{0}},\ket{\mathbf{1}},\ket{\mathbf{2}},\ket{\mathbf{3}}\}$ denotes two composed qubits ($a_{1}$ and $a_{2}$) of the first party (say Alice), which is further equivalent to a qu-quad system. 

Now, if Alice performs the measurement $\mathcal{N}_A\equiv\{N_1^A:=P[\ket{\mathbf{0+2}}_A],N_2^A:=P[\ket{\mathbf{0-2}}_A],N_3^A:=P[\ket{\mathbf{1+3}}_A],N_4^A:=P[\ket{\mathbf{1-3}}_A]\}$ on her qu-quad systems, she can easily distinguish the shared state depending upon the outcomes obtained by Bob and Charlie (see Table \ref{tab1}).

Further, by tracing out Alice's first qubit ($a_{1}$) both the states $\ket{\eta_{0}(\pm)}$ will be reduced to $\frac{1}{2}(\ketbra{000+111}{000+111}_{a_{2}bc}+\ketbra{000-111}{000-111}_{a_{2}bc})$, while tracing out her second qubit ($a_{2}$) produces $\frac{1}{2}(\ketbra{+}{+}_{a_{1}}\otimes\ketbra{00}{00}_{bc}+\ketbra{-}{-}_{a_{2}}\otimes\ketbra{11}{11}_{bc})$ both from $\ket{\eta_{0}(+)}$ and $\ket{\eta_{3}(-)}$. This guarantees that the set $\mathcal{S}_{5}^{(3)}$ is free from local redundancy.

Now, we are at the position to obtain tripartite GHZ basis in $\mathbb{C}^{2^{\otimes3}}$ from the set $\mathcal{S}_{5}^{(3)}$ under OPLM. If Alice performs the measurement $\mathcal{K}_A\equiv\{K_1^A:=P[(\ket{\mathbf{0}},\ket{\mathbf{1}})_A],K_2^A:=P[(\ket{\mathbf{2}},\ket{\mathbf{3}})_A]\}$
on her qu-quad system, then for the click $K_{1}^{A}$ the states will be converted to
\begin{equation*}
    \ket{\eta_{k}(\pm)}\to\ket{\phi_{k}(\pm)}=(\ket{\mathbf{0},k}\pm\ket{\mathbf{1},3-k})_{a_{1}a_{2}BC},~k\in\{0,1,2,3\}
\end{equation*}
and similarly for the click $K_{2}^{A}$, they will become
\begin{equation*}
    \ket{\eta_{k}(\pm)}\to\ket{\phi_{k}^{'}(\pm)}=(\ket{\mathbf{2},k}\mp\ket{\mathbf{3},3-k})_{a_{1}a_{2}BC},~k\in\{0,1,2,3\}.
\end{equation*}
Further, applying the unitary $\mathbb{U}_{3}=(\ketbra{2}{0}+\ketbra{3}{1}+\ketbra{0}{2}-\ketbra{1}{3})$ Alice can convert the states again in the standard GHZ-basis, which is known to be locally irreducible and indistingushable in any bipartition.
\begin{table}[h]
\centering
\begin{tabular}{ c|c|c|c|c|c }
&&\multicolumn{4}{c}{Bob-Charlie}\\
\hline
&& 00 & 01 & 10 & 11 \\
\hline
\multirow{4}{3em}{Alice} & $N_{1}^{A}$ & $\ket{\eta_{0}(+)}$ & $\ket{\eta_{1}(+)}$ & $\ket{\eta_{2}(+)}$ & $\ket{\eta_{3}(+)}$\\
 & $N_{2}^{A}$ & $\ket{\eta_{0}(-)}$ & $\ket{\eta_{1}(-)}$ & $\ket{\eta_{2}(-)}$ & $\ket{\eta_{3}(-)}$\\ 
 & $N_{3}^{A}$ & $\ket{\eta_{3}(-)}$ & $\ket{\eta_{2}(-)}$ & $\ket{\eta_{1}(-)}$ & $\ket{\eta_{0}(-)}$\\ 
  & $N_{4}^{A}$ & $\ket{\eta_{3}(+)}$ & $\ket{\eta_{2}(+)}$ & $\ket{\eta_{1}(+)}$ & $\ket{\eta_{0}(+)}$ \\ 
\hline
\end{tabular}
\caption{The table depicts LOCC distinguishability of the set $\mathcal{S}_{5}^{(3)}$}
\label{tab1}
\end{table}
\end{document}